\pdfoutput=1

\newif\ifDraft
\newif\ifSubmission
\newif\ifProceeding
\newif\ifFull
\newif\ifArxiv
\newif\ifJournal

\Draftfalse
\Submissionfalse
\Proceedingfalse
\Arxivtrue
\Journalfalse
\ifArxiv
\Fulltrue
\else\ifJournal
\Fulltrue
\else
\Fullfalse
\fi\fi

\documentclass{llncs}
\usepackage{hyperref,mathptmx,graphicx,amssymb,amsmath,verbatim}
\usepackage{times}
\usepackage{plaatjes}
\usepackage{color}

\graphicspath{{figures/}{./}}

\ifJournal
\ifSubmission
\usepackage{lineno}
\let\oldendproof\endproof
\def\endproof{\hfill \ensuremath{\Box}\oldendproof}
\fi
\else
\let\oldendproof\endproof
\def\endproof{\qed\oldendproof}
\fi

\newif\ifUseAppendix
\ifProceeding
\ifSubmission
\UseAppendixtrue
\else
\UseAppendixfalse
\fi
\else
\UseAppendixfalse
\fi

\newif\ifInAppendix
\newif\ifDeferToAppendix

\ifUseAppendix
\InAppendixfalse
\DeferToAppendixtrue  
\else
\InAppendixfalse
\DeferToAppendixfalse 
\fi

\ifDraft
\newcounter{mycommentcounter}

\newcommand{\Comment}[2][Comment]{\refstepcounter{mycommentcounter}%
    \ifhmode%
     \unskip%
     {\dimen1=\baselineskip \divide\dimen1 by 2 %
       \raise\dimen1\llap{\tiny
  {-\themycommentcounter-}}}\fi%
     \marginpar[{\renewcommand{\baselinestretch}{0.8}%
       \hspace*{-2em}\begin{minipage}{12em}\footnotesize%
[\themycommentcounter]:%
\raggedright \underline{#1}: #2\end{minipage}}]{\renewcommand{\baselinestretch}{0.8}%
       \begin{minipage}{12em}\footnotesize%
[\themycommentcounter]: \raggedright%
\underline{#1}: #2\end{minipage}}%
}
\else
\newcommand{\Comment}[2][Comment]{\relax}        
\fi

\newcommand{\marrow}{\marginpar[\hfill$\longrightarrow$]{$\longleftarrow$}}
\ifDraft
\newcommand{\personalremark}[3]{\textcolor{red}{\textsc{#1 #2:} \marrow}\textcolor{blue}{\textsf{#3}}}
\else
\newcommand{\personalremark}[3]{\relax}     
\fi

\newcommand{\maarten}[2][says]{\personalremark{Maarten}{#1}{#2}}

\newcommand{\define}[1]{{\bfseries\itshape #1}}

\title{Planar and Poly-Arc Lombardi Drawings}

\author{Christian A. Duncan\inst{1} \and
  David Eppstein\inst{2} \and
  Michael T. Goodrich\inst{2} \and \\
  Stephen G. Kobourov\inst{3} \and
  Maarten L{\"o}ffler\inst{2}}

\institute{\noindent
\inst{1}Department of Computer Science, Louisiana Tech Univ., Ruston, Louisiana, USA\\
\inst{2}Department of Computer Science, University of California, Irvine, California, USA\\
\inst{3}Department of Computer Science, University of Arizona, Tucson, Arizona, USA}

\begin{document}

\maketitle

\begin{abstract}
In Lombardi drawings of graphs, edges are represented as circular arcs,
and the edges incident on vertices have 
perfect angular resolution. However, not every graph has a Lombardi drawing, and not every planar graph has a planar Lombardi drawing. We introduce $k$-Lombardi drawings, in which each edge may be drawn with $k$ circular arcs, noting that every graph has a smooth $2$-Lombardi drawing. We show that every planar graph has a smooth planar $3$-Lombardi drawing and further investigate topics connecting planarity and Lombardi drawings.
\end{abstract}

\ifJournal
\ifSubmission
\linenumbers
\fi
\fi

\section{Introduction}

Motivated by the work
of the American abstract artist Mark Lombardi~\cite{lh-mlgn-03}, 
who specialized in drawings that illustrate financial and political networks,
Duncan {\it et al.}~\cite{degkn-dtwpa-10,degkn-ldg-10}
proposed a graph visualization called \emph{Lombardi drawings}.
These types of drawings attempt
to capture some of the visual aesthetics used by Mark Lombardi,
including his use of circular-arc edges and well-distributed
edges around each vertex.

A vertex with circular arc edges extending
from it has \define{perfect angular resolution} if the angles between
consecutive edges, as measured by the tangents to the circular arcs at the
vertex, all have the same degree.  A \define{Lombardi drawing} of a
graph $G=(V,E)$ is a drawing of a graph where every vertex is
represented as a point, the edges incident on each vertex have
perfect angular resolution, and
every edge is represented as a line segment or circular arc between
the points associated with adjacent vertices.

One drawback of previous work on Lombardi drawings is that (as we
prove here) not every graph has a Lombardi drawing. In this paper we
attempt to remedy this by considering drawings in which edges are
represented by multiple circular arcs. This added generality
allows us to draw any graph.

\paragraph{$k$-Lombardi Drawings.}
We define a \define{$k$-Lombardi drawing} to be a drawing with at most
$k$ circular arcs per edge, with a 1-Lombardi drawing being equivalent
to the earlier definition of a Lombardi drawing.
We say that a $k$-Lombardi drawing is \define{smooth} if every edge is
continuously differentiable, i.e., no edge in the drawing has a sharp
bend.
If a $k$-Lombardi drawing is not smooth, we say it is \define{pointed}.
Fortunately,
we do not need large values of $k$ to be able to draw all graphs: as we show, every graph has a smooth 2-Lombardi drawing.
Interestingly, this result is hinted at in the work of Lombardi 
himself---Figure~\ref{fig:lombardi} shows a portion of a drawing by Lombardi
that uses smooth edges consisting of two near-circular arcs.

\begin{figure}[hbt!]
  \centering
  \includegraphics[width=3in]{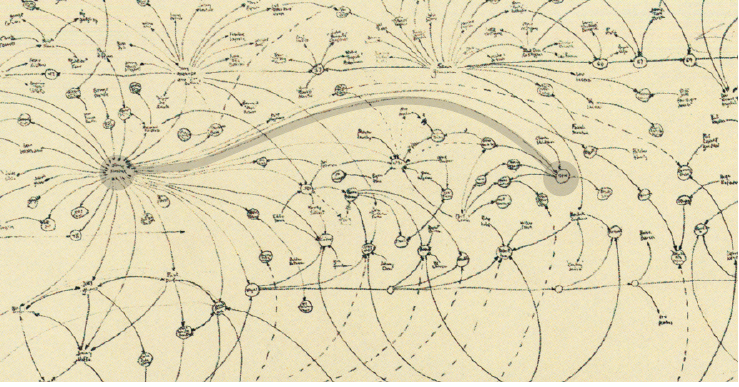} 
  \caption{
  A portion of Mark Lombardi, \textit{Chicago Outfit and Satellite Regimes,
  ca.~1931--83}, 1998,
  $48.125 \times 96.6225$ inches (cat.~no.~11)~\cite{lh-mlgn-03}.
Note the highlighted smooth two-arc edge.
  }
  \label{fig:lombardi}
\end{figure}


\begin{figure}[bht]
  \centering
  \includegraphics[width=2.75in]{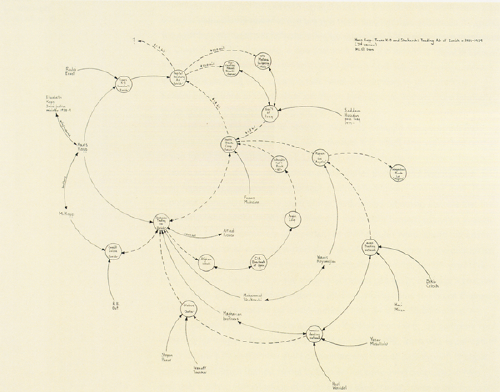} 
  \caption{Mark Lombardi, \textit{Hans Kopp, Trans K-B and Shakarchi
  Trading AG of Zurich, ca.~1981--89} (3rd Version), 1999,
  $20.25 \times 30.75$ inches (cat.~no.~22)~\cite{lh-mlgn-03}.}
  \label{fig:lombardi-planar}
\end{figure}

\paragraph{Planar Lombardi Drawings.}
Drawing planar graphs without crossings is a natural goal for graph
drawing algorithms and is easily achieved when angular
resolution is ignored. 
Lombardi himself avoided crossings in many of his drawings,
as shown in Figure~\ref{fig:lombardi-planar}. 
In previous work on Lombardi drawings,
Duncan {\it et al.}~\cite{degkn-ldg-10} 
showed that there exist embedded planar graphs that have Lombardi
drawings but do not have {\em planar Lombardi drawings}.
Here we continue this investigation of planar Lombardi drawings and extend it to planar $k$-Lombardi drawings.

\paragraph{New Results.}
In this paper we provide the following results:
\begin{enumerate}
\item We find examples of graphs that do not have a Lombardi drawing,
  regardless of the ordering of edges around each vertex, thus
  strengthening an example from \cite{degkn-ldg-10} of graphs for
  which a specific edge ordering cannot be drawn.
\item We show how to construct a smooth 2-Lombardi drawing for any graph.
\item  We
find examples of planar 3-trees with no planar Lombardi
drawing, strengthening an example from  \cite{degkn-ldg-10} of a
planar graph with treewidth greater than three that is not planar Lombardi.
\item We show how to represent any planar graph of maximum degree three with a smooth 2-Lombardi planar drawing and any planar graph with a pointed 2-Lombardi planar drawing or a smooth 3-Lombardi planar drawing.
\end{enumerate}

\paragraph{Other Related Work.}
In addition to the earlier work on Lombardi drawings, there is
considerable prior work on 
graph drawing with circular-arc or curvilinear edges for the sake of
achieving good, but not necessarily perfect, angular resolution~\cite{cdgk-dpg-01,gw-fdpg-01}.
There is also significant work on
\emph{confluent
drawings}~\cite{%
DBLP:journals/jgaa/DickersonEGM05,%
DBLP:conf/gd/EppsteinGM05,%
DBLP:journals/algorithmica/EppsteinGM07,%
hmr-becg-07,hv-fdeb-09}, which use curvilinear edges not to separate edges but rather to bundle similar edges together and avoid edge crossings.
Brandes and Wagner~\cite{bw-uglv-98} provide a force-directed algorithm
for visualizing train schedules using B{\'e}zier curves for edges and
fixed positions for vertices.
Finkel and Tamassia~\cite{ft-cgduf-05} extend this work by giving
a force-directed
method for drawing graphs with curvilinear edges where vertex
positions are not fixed. 
Aichholzer {\em et al.}~\cite{aaadj-at-10} show,
for a given embedded planar triangulation with fixed vertex positions,
it is possible to find a circular-arc 
drawing that maximizes the minimum angular resolution 
by solving a linear program.
In addition,
Matsakis~\cite{m-trgd-10} describes a force-directed approach to
producing Lombardi drawings, but without an implementation.
Goodrich and Trott~\cite{gt-fdls-11} and
Chernobelskiy
{\em et al.}~\cite{cck-lse-11}, on the other hand,
describe functional Lombardi force-directed schemes, which are respectively 
based on the use of dummy vertices and tangent forces, but may not
always achieve perfect angular resolution.
\ifFull
Interestingly,
Efrat {\em et al.}~\cite{eek-flc-07} show that, given a fixed placement 
of the vertices of a planar graph, it is NP-complete
to determining whether the edges can be drawn with circular 
arcs so that there are no crossings. 
\fi
Thus,
to the best of our knowledge, none of this other related work correctly results in drawings
of graphs having perfect angular resolution and curvilinear edges.

Alternatively, some previous  work 
achieves good angular resolution using straight-line 
drawings~\cite{dv-aptg-96,gt-pdara-94,mp-arpg-94}
or piecewise-linear poly-arc drawings~\cite{elmn-o3arldg-10,gm-ppdga-98,k-dpguc-96}.
Di~Battista
and Vismara~\cite{dv-aptg-96} 
characterize straight-line drawings of planar graphs
with a prescribed assignment of angles between consecutive edges
incident on the same vertex.

\section{$k$-Lombardi Drawings}
\label{sec:kLombardi}
In this section, we investigate $k$-Lombardi drawings.
First, we establish the need to use poly-arc edges in order to
be able to draw any graph.

\subsection{Non-Lombardi Graphs}

\ifInAppendix
\else
Duncan {\em et al.}~\cite{degkn-ldg-10} show a graph, Figure~\ref{fig:G7_bad}, for
which no Lombardi drawing is possible 
{\em while preserving the given ordering of edges around each vertex}.
However, as Figure~\ref{fig:G7_good} shows, if the ordering is not fixed,
it is possible to create a valid Lombardi drawing for the graph.
In this section, we provide a graph that has no Lombardi drawing {\em irrespective of the edge ordering}.

There are some complications in proofs of non-Lombardi counterexamples that differ from
counterexamples in straight-line planar drawings.
For example, if graph $G$ is non-Lombardi, this does not imply
that all graphs $H \supset G$ are non-Lombardi because the addition
of edges changes the angular resolution and can therefore
dramatically change the subsequent placement of vertices.
In addition, since the edge ordering is not fixed by the input, we
must argue that any ordering forces a conflict.

Additional complications concern the density and symmetry of any possible counterexample.
A \define{$k$-degenerate graph} is a graph that can be 
reduced to the empty graph by iteratively removing vertices 
of degree at most $k$.
The graph in Figure~\ref{fig:G7_bad+G7_good} is 3-degenerate, and 3-degenerate graphs can be drawn
Lombardi-style if we are willing to ignore vertex-vertex and vertex-edge overlaps.\footnote {Note that a drawing with vertex-vertex overlaps would still need to obey the perfect angular resolution constraints on the (possibly zero-length) edges.}
\maarten {Added this footnote in response to a reviewer comment; not sure if it actually makes things clearer.}
Consequently, if a 3-degenerate graph is to be a counterexample, we must show that all 
vertex orderings force two vertices to overlap.
Intuitively, 4-degenerate graphs should be more restrictive,
but the simplest 4-degenerate graph, $K_5$, nevertheless has a circular Lombardi drawing.
One issue is the fact that $K_5$ is extremely symmetrical.
Therefore, we shall modify this graph to break its symmetry.
We define our counterexample graph $G_8$ to be $K_5$ with the addition of three
degree-one vertices causing one of the vertices
of the original $K_5$ to have degree 5 and another to have degree 6, while
the other three remain with degree 4;
see Figure~\ref{fig:G8_simple}.

\tweeplaatjes[height=1.25in]{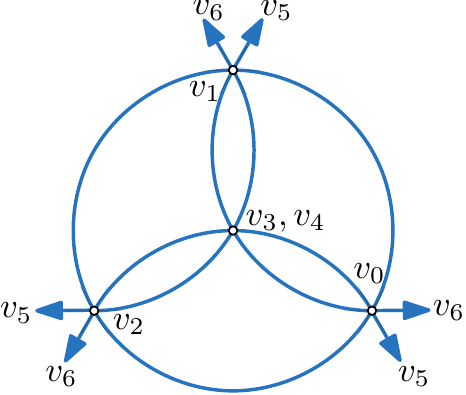} {G7_good} {A 7-vertex 3-degenerate graph that has no Lombardi drawing with the given edge ordering.
  (a) A M\"obius transformation makes triangle $v_0v_1v_2$ equilateral, 
  forcing vertices $v_3$ and $v_4$ to both be placed at the centroid and vertices $v_5$ and $v_6$ at the point at infinity; 
  (b) A different ordering that does provide a Lombardi drawing.}

Before we can establish our main theorem, we need to present
a few geometric properties related to Lombardi drawings.
\begin{property}[\cite{degkn-ldg-10}]
\label{prop:circleAngle}
Let $A$ be a circular arc or line segment connecting two points $p$ and $q$ that both lie on circle $O$. Then $A$ makes the same angle to $O$ at $p$ that it makes at $q$.
Moreover, for any $p$ and $q$ on $O$ and any angle $0 \le\theta\le \pi$, there exist either two arcs or a line segment and pair of collinear rays connecting $p$ and $q$, making angle $\theta$ with $O$, one lying inside and one outside of $O$.
\end{property}
\fi
\ifUseAppendix
\ifDeferToAppendix
We defer the proof of the next property, partially established in~\cite{degkn-ldg-10}, to the appendix.
\fi
\else
The next property was partially established in~\cite{degkn-ldg-10}.
\fi
\ifInAppendix
We now present the proof for Property~\ref{prop:arcLocus}.
\else
\begin{property}
\label{prop:arcLocus}
Suppose we are given two points $p=(p_x,p_y)$ and $q=(q_x,q_y)$ 
and associated
angles $\theta_{ph}$ and $\theta_{qh}$ and an angle $\theta_{pq}$.
Consider all pairs of circular arcs 
that leave $p$ and $q$ with angles $\theta_{ph}$ and $\theta_{qh}$
respectively (measured with respect to the positive horizontal axis)
and meet at an angle $\theta_{pq}$.
The locus of meeting points for these pairs of arcs is a circle.
Moreover, the circle has radius $r_c=d_{pq}\csc\alpha/2$ 
and center $(p_x + r_c \sin(\alpha+\beta), p_y - r_c \cos(\alpha + \beta))$.
where $\alpha = (\theta_{ph} - \theta_{qh} - \theta_{pq})/2$,
$\beta$ is the angle formed by the ray from $p$ through $q$
with respect to the positive horizontal axis, and $d_{pq}$ is the
distance between the points $p$ and $q$.
\end{property}
\fi
\ifDeferToAppendix
\else
\ifUseAppendix
\begin{proof}[Property~\ref{prop:arcLocus}]
\else
\begin{proof}
\fi
See~\cite{2010arXiv1009.0579D} for the initial details.
For simplicity at the moment, let us assume that $p$ and $q$
are aligned horizontally, that is $\beta=0$.
Let $C$ represent the circular locus with 
center $c=(x_c, y_c)$ and radius $r_c$.
From~\cite{2010arXiv1009.0579D} we know that the angle formed by
the center of the circle and the two points $q$ and $p$ is
$\angle qcp = \theta_{ph} - \theta_{qh} - \theta_{pq} = 2\alpha$.
Analyzing the isoceles triangle $\triangle QPC$, 
we determine the radius $r_c = d_{pq}/(2\sin\alpha)$.

Now, if $\beta \neq 0$, 
a simple rotation of $-\beta$ about $p$ can be applied yielding
$\alpha = \theta_{ph} - \beta - \theta_{qh} + \beta - \theta_{pq}$
and hence the angle $\alpha$ and the radius $r_c$ are unaffected.

Using basic trigonometry and geometry, 
we can also determine the center of this circle as 
$c=(p_x + r_c \sin(\alpha+\beta), p_y - r_c \cos(\alpha + \beta))$.
\end{proof}
\fi

\ifInAppendix
\else
\begin{theorem}
\label{thm:G8}
The graph $G_8$ is non-Lombardi.
\end{theorem}

\begin{proof}
Let $v_0, v_1, v_2$ be the three vertices of $G_8$ with degree four.
Let $v_3$ and $v_4$ be the vertices with degree five and six 
respectively.
We do not care about the final placement of the degree-one vertices,
whose main purpose is to alter the angular resolution of $v_3$
and $v_4$.
Using a M\"obius transformation we can assume that the first
three vertices $v_0$, $v_1$, and $v_2$ are placed
on the corners of a unit equilateral triangle such that $v_0$ and $v_1$ have positions $(0,0)$ and $(1,0)$ respectively.
We shall show that for every edge ordering, the
two vertices $v_3$ and $v_4$ cannot both be placed to
maintain correctly their angular resolution and be connected to
each other.
We do this by establishing the algebraic equations for their positions
based on the edge orderings of all vertices.
We then show that such a set of equations has no solution for
any valid assignment of orderings.

We first establish a notation for representing a specific edge
ordering.
For every vertex $v_i$ with neighbor $v_j$, let $k_{ij}$ 
represent the counterclockwise
cyclic ordering of edge $(v_i,v_j)$ about $v_i$
with $k_{01}=0$ and $k_{i0}=0$ for $i > 0$.
For example, in Figure~\ref{fig:G8_simple}, the edge ordering
around $v_4$ has $k_{41}=2$, $k_{42}=4$, $k_{43}=5$,
$k_{46}=1$, and $k_{47}=3$.
The \define{twist} $t_i$ of a vertex $v_i$ is the angle made by the 
arc extending from $v_i$ to the neighbor $v_j$ with $k_{ij}=0$.
From the initial placement of $v_0$, $v_1$, and $v_2$ on an 
equilateral triangle and their 
respective edge orderings, we can uniquely determine the twists for 
each of these vertices;
see Figure~\ref{fig:twistComp}.
Since the three vertices lie on an equilateral triangle, the
tangents to the circle defined by the three points also form an
equilateral triangle.
From Property~\ref{prop:circleAngle}, the angles formed by the
arcs connecting each pair of vertices to the tangents at the circle
yield matching (but undetermined) angles, labeled $a$, $c$, and $e$.
The angles $b$, $d$, and $f$ are determined uniquely by the edge 
orderings as follows:
\ifProceeding
\begin{equation}
b = 2\pi - k_{02} \pi/2, \hspace{2em}
d = k_{12} \pi/2, \hspace{2em}
f = 2\pi - k_{21} \pi/2\label{eqn:twF}
\end{equation}
\else
\begin{align}
b & = 2\pi - k_{02} \pi/2\label{eqn:twB}\\
d & = k_{12} \pi/2\label{eqn:twD}\\
f & = 2\pi - k_{21} \pi/2\label{eqn:twF}
\end{align}
\fi
Noting that certain triplets of angles yield a value of $\pi$,
we have the following three equations on three unknowns:
\ifProceeding
$a + b + c = \pi + 2i_0\pi$,
$c + d + e = \pi + 2i_1\pi$, and
$e + f + a = \pi + 2i_2\pi$.
\else
\begin{align}
a + b + c & = \pi + 2i_0\pi\\
c + d + e & = \pi + 2i_1\pi\\
e + f + a & = \pi + 2i_2\pi.
\end{align}
\fi
Solving for $a$ yields: $2a = \pi - f - b + d + 2(i_0 - i_1 + i_2)\pi$.
For the twist for $v_0$, we wish to know the value of $x$,
the angle for the arc from $v_0$ to $v_1$.
Noting that $x = a + b + 2\pi/3 - 2i_0\pi$ and substituting in 
\ifProceeding
Equation (\ref{eqn:twF})
\else
Equations (\ref{eqn:twB}-\ref{eqn:twF}) 
\fi
yields 
$t_0 = x = 7\pi/6 + \pi (k_{12} + k_{21} - k_{02})/4 + (i_2 - i_0 - i_1)\pi$.
Noting that $t_0 + c + \pi/3 = 2\pi$ yields $t_1 = \pi - t_0$.
Similarly, $t_2 = \pi - a = 5\pi/3 - t_0 - k_{02}\pi/2 + 2\pi(1-i_0)$.

\tweeplaatjes[height=1.25in] {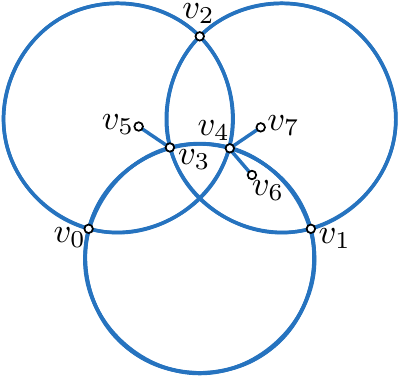} {twistComp}
{(a) $G_8$ with $K_5$ part drawn Lombardi-style and additional
edges shown. (b) Computing the twist for the three vertices $0$, $1$, and $2$.
The twist for vertex $0$ is $x$.}

The positions and orienting twists of the first three vertices also 
yield a unique position and twist for vertices $v_3$ and $v_4$.
After determining these values, we shall show that in all orderings
it is not possible to connect $v_3$ to $v_4$ with a single circular
arc while still maintaining the proper angular resolution.

From Property~\ref{prop:arcLocus}, $v_3$ must lie on a circle $C_{01}$
defined by the neighbors $v_0$ and $v_1$ and their corresponding
arc tangents.
Similarly, it must lie on circles $C_{02}$ and $C_{12}$.
The intersection of these three circles determines the position 
and orientation of $v_3$.
Let us proceed to determine $C_{01}$.
Letting $p=v_0$ and $q=v_1$,
we have $\theta_{ph}=t_0+\pi k_{03}/4$ and 
$\theta_{qh}=t_1 + \pi k_{13}/4$ and 
$\theta_{pq}=\pi (k_{31}-k_{30})/5 = \pi k_{31}/5$.
From Property~\ref{prop:arcLocus} and the fact that $d_{pq}=1$, 
we can determine that $C_{01}$ has radius
$r_{01} = \csc\alpha_{01}/2$ and center
$c_{01} = (r_{01} \sin\alpha_{01}, -r_{01}\cos\alpha_{01}) =
(1/2, -\cot\alpha_{01}/2)$
with $\alpha_{01} = (\theta_{ph}-\theta_{qh}-\theta_{pq})/2 =
t_0 - \pi/2 + \pi(5k_{03} - 5k_{13} - 4k_{31})/40$.
Similarly, $C_{02}$ has radius $r_{02}=\csc\alpha_{02}/2$
and center $c_{02} = (r_{02}\sin(\alpha_{02}+\pi/3),
-r_{02}\cos(\alpha_{02}+\pi/3))$ with 
$\alpha_{02} = t_0 - 5\pi/6 + \pi(5k_{03} + 10k_{02} - 5k_{23} - 4k_{32})/40 + (i_0-1)\pi$.

Given the circles and the position of $v_0$ at the origin,
it is easy to determine the intersection of the two circles,
one of which is $v_0$ and the other, if it even exists,
must be $v_3$.
Since $v_0$ must lie on the intersection, the line from $v_0$
to $v_3$ is perpendicular to the line, $\ell$, through the two
centers.
Moreover, $v_3$ is the reflection of $p$ about $\ell$.
Thus, letting $\vec{v} = (v_x, v_y) = c_{02} - c_{01}$,
$\vec{c} = v_0 - c_{01} = -c_{01}$, and
$\vec{v}^\perp = (-v_y, v_x)$ yields
\ifProceeding
$v_3 = \frac{-2\vec{c} \cdot \vec{v}^\perp}{\vec{v}\cdot\vec{v}}\vec{v}^\perp.$
\else
\begin{align}
v_3 & = \frac{-2\vec{c} \cdot \vec{v}^\perp}{\vec{v}\cdot\vec{v}}\vec{v}^\perp.
\end{align}
\fi
To establish the twist $t_3$ at $v_3$ we observe from 
Property~\ref{prop:circleAngle} that the angle 
$\alpha$ formed by the line $\ell_{03}$ from $v_0$ to $v_3$
and the tangent of the curve from $v_0$ to $v_3$ is the same as 
the tangent of the curve from $v_3$ to $v_0$ and the line $\ell_{03}$.
Moreover, $\theta_{03} = t_0 + k_{03}\pi/4 = \alpha + \beta_{03}$
and $t_3 = \theta_{30} = \pi - \alpha + \beta_{03}$ where
$\beta_{03} = \arctan(v_3(y)/v_3(x))$ is the slope of $\ell_{03}$.
From this, we can deduce that 
$t_3 = \pi - t_0 - k_{03}\pi/4 + 2\beta_{03}$.
The exact same calculations can be used to compute $v_4$ and $t_4$.

As with the twists for $t_3$ and $t_4$, we can use 
Property~\ref{prop:circleAngle} to determine the angles formed by
the arc from $v_3$ to $v_4$ given their positions and twists.
We know that the angles of the tangents to the arc
at $v_3$ and $v_4$ are $\theta_{34} = t_3 + k_{34}\pi/5$
and $\theta_{43} = t_4 + k_{43}\pi/6$ respectively.
Letting $\beta_{34}= \arctan((v_4(y) - v_3(y))/(v_4(x) - v_3(x)))$ 
be the slope of the line from $v_3$ to $v_4$, we have
that $\theta_{34} - \beta_{34} = \alpha$ and
$\pi-\alpha = \theta_{43} - \beta_{34}$.
Consequently, we have
\begin{align}
\theta_{34} + \theta_{43} & = \pi + 2\beta_{34}\label{eqn:contra}.
\end{align}

Each specific edge ordering therefore yields a unique 
set of positions and twists for $v_3$ and $v_4$ as outlined above.
To show that no Lombardi drawing is possible one must
simply show that Equation~\ref{eqn:contra} does not hold
for {\em any} edge ordering.
Though there are a finite number of possible orderings and
though symmetries could be used to reduce that number, the
individual case analysis for such a proof appears to be quite
unwieldy.
Instead, we simply iterate over every possible edge ordering,
applying these equations to a numerical algorithm
that searches for a valid non-contradictory assignment.
The Python code for this program is shown 
in 
\ifUseAppendix
the Appendix in
\fi
Table~\ref{code:nonLombardiCode}.
By running this program, one can see that no valid assignments are
possible concluding our proof.
\end{proof}

\begin{corollary}
\label{cor:infNon}
There are an infinite amount of connected non-Lombardi graphs.
\end{corollary}
\fi

\ifDeferToAppendix
See the Appendix for details on the proof.
\else
\ifUseAppendix
The following is the proof of Corollary~\ref{cor:infNon}:
\begin{proof}[Corollary~\ref{cor:infNon}]
\else
\begin{proof}
\fi
Let $G$ be formed from a graph $G'$, having at least two degree-one 
vertices $u$ and $v$ that do not share a common neighbor,
by merging $u$ and $v$ and
creating a degree-two vertex $w$.
If $G$ is Lombardi, then so is $G'$ as we can take a Lombardi
drawing of $G$, split $w$, place $u$ and $v$ on the arcs
between $w$ and its respective neighbor, and still maintain
a valid Lombardi drawing.
Thus, we can take any collection of disjoint copies of $G_8$
and combine degree-one vertices 
to form a connected non-Lombardi graph.
\end{proof}
\fi

\ifDeferToAppendix
\else
\ifInAppendix
Table~\ref{code:nonLombardiCode} presents the complete 
code in Python for testing that no edge orderings of $G_8$ have a
valid assignment.
\fi
\begin{table}
\scriptsize
\begin{verbatim}
#!/usr/bin/python

from itertools import *
from bigfloat import *

def match(k0,k1,k2,k3,k4,i0,i1,i2):
    (k01,k02,k03,k04)=(0,k0[0],k0[1],k0[2])
    (k10,k12,k13,k14)=(0,k1[0],k1[1],k1[2])
    (k20,k21,k23,k24)=(0,k2[0],k2[1],k2[2])
    (k30,k31,k32,k34)=(0,k3[0],k3[1],k3[2])
    (k40,k41,k42,k43)=(0,k4[0],k4[1],k4[2])
    b,d,f = 2 - k02/2.0, k12/2.0, 2-k21/2.0            # Eqs 1-3
    t0 = 7.0/6.0 + (k12 + k21 - k02)/4.0 + (i2-i0-i1)  # The twists

    # Compute v3 and t3
    a01 = t0 - 0.5 + (5*k03 - 5*k13 - 4*k31)/40.0
    a02 = t0 + i0-11.0/6.0 + (5*k03 + 10*k02 - 5*k23 - 4*k32)/40.0
    r01, r02 = 0.5/sin(a01 * const_pi()), 0.5/sin(a02 * const_pi())
    c01 = (0.5, -0.5/tan(a01*const_pi()))
    c02 = (r02*sin((a02 + 1.0/3.0)*const_pi()), -r02*cos((a02 + 1.0/3.0)*const_pi()))
    v = (c02[0] - c01[0], c02[1] - c01[1])
    M = 2.0 * (c01[1] * v[0] - c01[0] * v[1])/(v[0]*v[0]+v[1]*v[1])
    v3 = (-v[1] * M, v[0] * M)
    b03 = atan2(v3[1], v3[0])/const_pi()
    t3 = 1 - t0 - k03/4.0 + 2*b03
    
    # Compute v4 and t4
    a01 = t0 - 0.5 + (3*k04 - 3*k14 - 2*k41)/24.0
    a02 = t0 + i0-11.0/6.0 + (3*k04 + 6*k02 - 3*k24 - 2*k42)/24.0
    r01, r02 = 0.5/sin(a01 * const_pi()), 0.5/sin(a02 * const_pi())
    c01 = (0.5, -0.5/tan(a01*const_pi()))
    c02 = (r02*sin((a02 + 1.0/3.0)*const_pi()), -r02*cos((a02 + 1.0/3.0)*const_pi()))
    v = (c02[0] - c01[0], c02[1] - c01[1])
    M = 2.0 * (c01[1] * v[0] - c01[0] * v[1])/(v[0]*v[0]+v[1]*v[1])
    v4 = (-v[1] * M, v[0] * M)
    b04 = atan2(v4[1], v4[0])/const_pi()
    t4 = 1 - t0 - k04/4.0 + 2*b04

    # Compare v3,t3 and v4,t4
    t34,t43 = t3 + k34/5.0, t4 + k43/6.0
    b34 = atan2(v4[1]-v3[1], v4[0]-v3[0])/const_pi()

    # Compare and account for small errors in round-off
    lhs, rhs = (t34 + t43, 1 + 2 * b34)
    diff = mod((lhs - rhs) if lhs > rhs else (rhs - lhs), 2)
    epsilon = 0.000001
    if (diff < epsilon or diff > 2-epsilon):
        return True    # Found a valid assignment

for k0 in permutations(range(1,4)):
    for k1 in permutations(range(1,4)):
        for k2 in permutations(range(1,4)):
            for k3 in permutations(range(1,5),r=3):
                for k4 in permutations(range(1,6),r=3):
                    for (i0,i1,i2) in product(range(0,2), repeat=3):
                        with precision(100):
                            if match(k0,k1,k2,k3,k4,i0,i1,i2):
                                print "Valid match found."

\end{verbatim}
\caption{Python code to verify $G_8$ is non-Lombardi}
\label{code:nonLombardiCode}
\end{table}
\fi

\subsection{Smooth 2-Lombardi Drawings}
If we want to draw Lombardi-style drawings for any given graph we have
to relax one of the two requirements that specify Lombardi drawings.
Ideally, we would like to avoid
relaxing the requirement that edges have perfect angular resolution.
Fortunately, we can achieve a 
Lombardi methodology for drawing any graph if
we allow two circular arcs per edge.

\maarten {I downgraded this to a corollary, given complaints from all reviewers.}

We recall from Duncan {\it et al.}~\cite[Theorem 3]{degkn-ldg-10}:
\textit {Every $2$-degenerate graph with a specified cyclic ordering of the edges around each vertex has a Lombardi drawing.}

\begin{corollary}
Every graph has a smooth 2-Lombardi drawing.
Furthermore, the vertices can be chosen to be in any fixed position.
\end{corollary}
\begin{proof}
Starting
with the given graph $G$, subdivide every edge by dividing it in two
and adding a ``dummy'' vertex incident to these two new edges.
By the above theorem, there exists a Lombardi drawing of the resulting $2$-degenerate graph $G_2$.
Furthermore, each dummy vertex of $G_2$ that was added to subdivide an edge
of $G$ has degree $2$; hence, in a Lombardi drawing of $G_2$ the
edges incident on each such dummy vertex have tangents that meet at 180
degrees.
Thus, when we consider these two circular arcs of $G_2$ as a single edge of
$G$ they define a smooth two-arc edge.
See Figure~\ref{fig:G8_2}.

The $2$-degenerate drawing algorithm orders the vertices in such a way that each vertex has at most two earlier neighbors; it places vertices with zero or one previous neighbor freely, but vertices with two previous neighbors are constrained to lie on a circular arc. For $G_2$, we can choose an ordering in which only the dummy vertices have two previous neighbors; therefore, the vertices of $G$ can have any initial placement.
\end{proof}

As Figure~\ref{fig:K4_2} illustrates,
although we can place the vertices in any position with any initial orientation,
an arc's smooth bend point might be an inflection point.

\tweeplaatjes[height=1.25in]{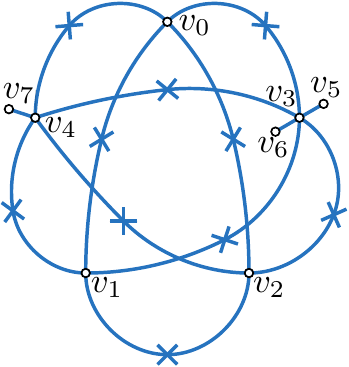} {K4_2} {(a) An example 2-Lombardi drawing of $G_8$.  The bend points (not all of which
are necessary) are shown with crossed marks.
(b) An example 2-Lombardi drawing of $K_4$ with the vertices placed on a line
and tangents oriented to force numerous inflection points.}

\section{Planar $k$-Lombardi Drawings}

In this section, we investigate \emph {planar} (non-crossing) Lombardi drawings and planar $k$-Lombardi drawings.

\subsection{A planar 3-tree with no planar Lombardi drawing}

It is known that planar graphs do not necessarily have planar Lombardi drawings.
For example, 
Duncan {\it et al.}~\cite{degkn-ldg-10}
show that the nested triangles graph must have edge crossings whenever there are 4 or more levels of nesting. While this graph is 4-degen\-er\-ate, even more constrained classes of planar graphs have no planar Lombardi drawings. Specifically, we can show that there exists a planar 3-tree that has no planar Lombardi realization. The planar 3-trees, also known as Apollonian networks, are the planar graphs that can be formed, starting from a triangle, by repeatedly adding a vertex within a triangular face, connected to the three triangle vertices, subdividing the face into three smaller triangles. These graphs have attracted much attention within the physics research community both as models of porous media with heterogeneous particle sizes and as models of social networks~\cite{AndHerAnd-PRL-05}. In addition, 3-trees are relevant for Lombardi drawings because they are examples of 3-degenerate graphs, which have nonplanar Lombardi drawings if vertex-vertex and vertex-edge overlaps are allowed.

\eenplaatje[width=3in]{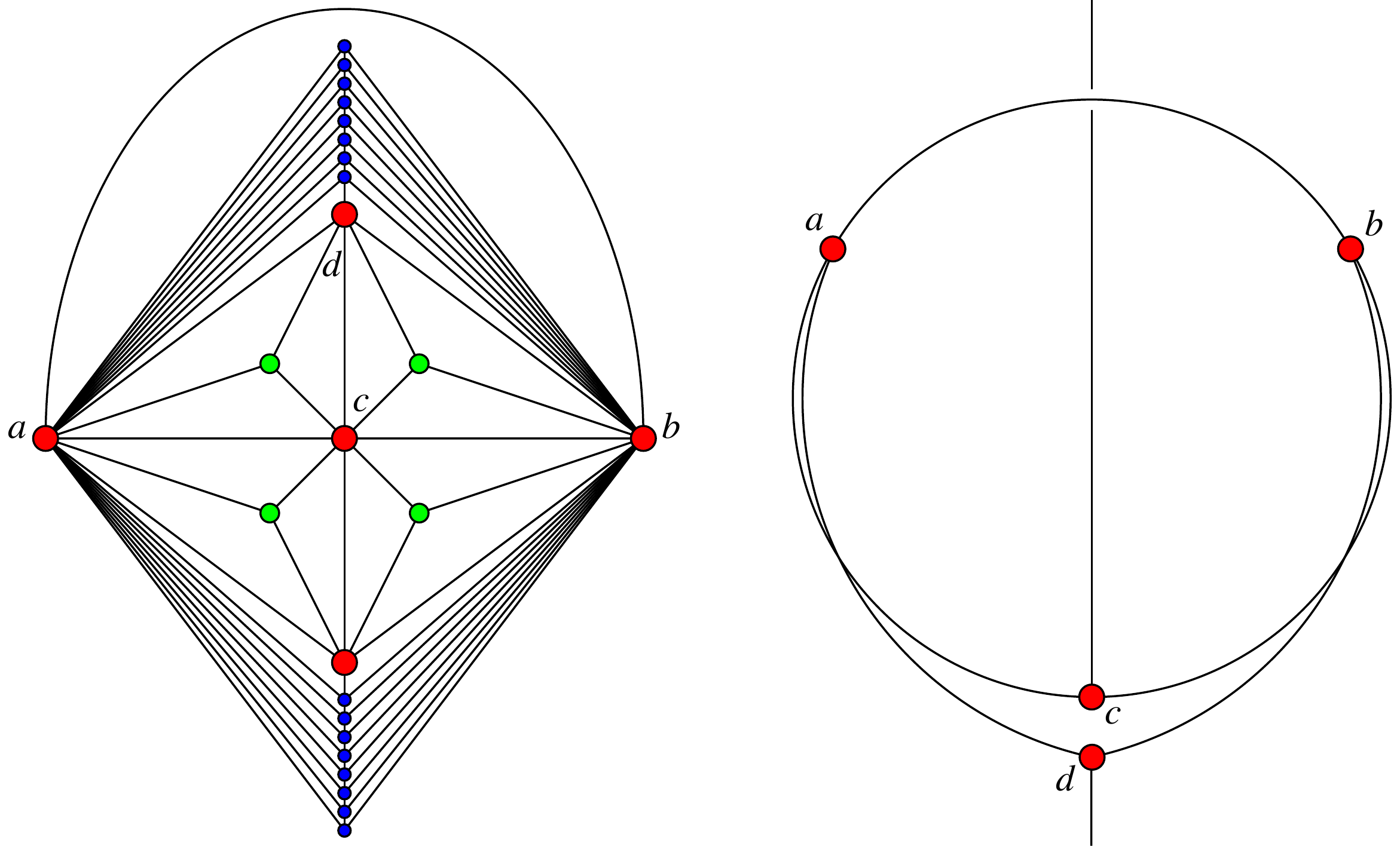}{Left: A planar 3-tree that has no planar Lombardi drawing. Right: For the $K_4$ subgraph defined by the four vertices $a$, $b$, $c$, and $d$, a drawing with the correct angles at each vertex will necessarily have crossings.}

\begin {theorem}
There exists a planar 3-tree that has no planar Lombardi drawing.
\end {theorem}

\begin {proof}
An example of a planar 3-tree that has no planar Lombardi drawing is given in Figure~\ref{fig:apollonian}; in the figure, sixteen small blue vertices are shown, but our construction requires a sufficient number (which we do not specify precisely) in order to force the angle between arcs $ad$ and $ab$ to be arbitrarily close to $180^\circ$. 
The numbers of blue vertices on the top and bottom of the figure should be equal.
Because of this equality, the three arcs $ab$, $bc$, and $ca$ split the graph into two isomorphic subgraphs, and due to this symmetry they must meet at $180^\circ$ angles to each other, necessarily forming a circle in any Lombardi drawing. By performing a M\"obius transformation on the drawing, we may assume without loss of generality that these three points form the vertices of an equilateral triangle inscribed within the circle, as shown in the right of the figure. Then, according to our previous analysis of 3-degenerate Lombardi graph drawing, there is a unique point in the plane at which vertex $d$ may be located so that the arcs $ad$, $bd$, and $cd$  form the correct $120^\circ$ angles to each other and the correct angles to the three previous arcs $ab$, $bc$, and $ca$. However, as shown on the right of the figure, that unique point lies outside circle $abc$ and causes multiple edge crossings in the drawing.
\end {proof}


\subsection{Planar $2$-Lombardi drawings for planar max-degree $3$ graphs}

We will show that planar graphs of maximum degree $3$ allow for smooth planar $2$-Lombardi drawings. 

\begin {lemma} \label {lem:threepoints}
Given a circle $C$ and three points $a$, $b$, and $c$ on it, there exists a point $p$ inside $C$ such that we can draw three edges from $p$ to $a$, $b$, and $c$  as circular arcs that are all perpendicular to $C$, and meet inside $p$ at $120^\circ$ angles.
\end {lemma}

\begin {proof}
We can find a M\"obius transform $\tau$ that maps the circle to itself, mapping $a$, $b$, and $c$  to three points $a'$, $b'$ and $c'$ that are $120^\circ$ apart on the circle. For these three points, the three edges can be drawn as radii of the circle meeting at the center point $p'$. The inverse transformation to $\tau$ maps $p'$ to $p$ and maps these three radii to circular arcs with the desired property.
\end {proof}

\begin {theorem}
Every planar graph with maximum degree three has a planar smooth $2$-Lombardi drawing. 
\end {theorem}

\drieplaatjes[scale=0.75]{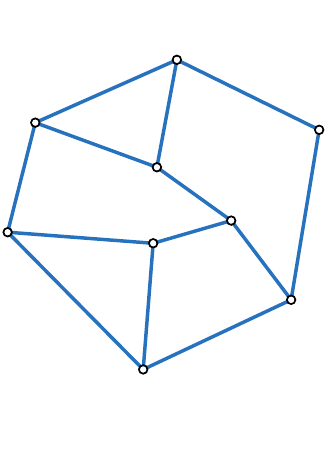}{deg3lom1-perp}{deg3lom1-output} {(a) A planar graph of maximum degree 3. (b) A representation of the graph as tangent circles according to the Koebe--Andreev--Thurston theorem, together with arcs connecting each vertex perpendicularly to the disk tangency points. Layout generated using Ken Stephenson's \href{http://www.math.utk.edu/~kens/CirclePack/}{CirclePack} software. (c) The final smooth $2$-Lombardi drawing. }

\begin {proof}
We apply the Koebe--Andreev--Thurston theorem to create a representation of the given graph as the intersection graph of tangent circles, as in Figure~\ref {fig:deg3lom1-perp}. Each circle has three contact points which will be the bend points of its incident edges. We apply Lemma~\ref {lem:threepoints} to the circles to obtain a vertex and half-edge drawing inside each disk. Since at each contact point two half-edges meet at an angle of $180^\circ$,
the result is a planar smooth $2$-Lombardi drawing of $G$. 
\end {proof}

\subsection{Planar $2$-Lombardi pointed drawings for planar graphs}

We now show that every planar graph allows a planar $2$-Lombardi drawing with pointed joints. The approach is similar to the previous section, but the drawing method inside the disks is different.
We need the following lemmas:

\begin {lemma} \label {lem:quadrants}
Let $C$ be a circle, and $P$ be a set of $n$ points on $C$. Additionally suppose that the four integers $n_1, n_2, n_3, n_4$ sum up to $n$ and satisfy the inequalities $\lfloor n/4\rfloor\le n_i\le\lceil n/4\rceil$ and $\lfloor n/2\rfloor\le n_i+n_{(i+1)\bmod 4}\le\lceil n/2\rceil$. Then there exist two circles $A$ and $B$ disjoint from $P$ such that $A$, $B$, and $C$ are pairwise perpendicular and such that $A$ and $B$ subdivide $P$ into four sets of cardinality $n_1$, $n_2$, $n_3$ and $n_4$.
\end {lemma}

\drieplaatjes[scale=0.75] {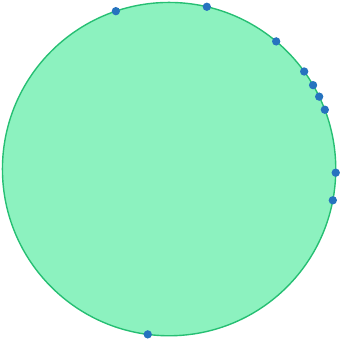} {plalom2-split} {plalom2-edges} {(a) A disk with a set of connection points on its boundary. (b) A placement for the vertex in the disk that divides the connection points into four quadrants. (c) The actual connections are not fixed, and guaranteed to not intersect.}

\ifFull

It is convenient to begin with a continuous analogue of the lemma. We define a \emph{smooth} probability distribution on $C$ to be a distribution that assigns a nonzero probability to any arc of $C$, such that arbitrarily short arcs have a probability that approaches zero.

\begin{lemma}
\label{lem:smooth}
Let $C$ be a circle, and $\Pi$ be a smooth probability distribution on $C$. Then there exist two circles $A$ and $B$ such that $A$, $B$, and $C$ are pairwise perpendicular and such that the four arcs of $C$ formed by its crossing points with $A$ and $B$ each have probability $1/4$ under distribution $\Pi$.
\end {lemma}

\begin{proof}
We may view $A$ and $B$ as arcs inside $C$ (ignoring part of the circles) that end perpendicular to $C$, and cross each other at a $90^\circ$ angle. Figure~\ref {fig:plalom2-split} illustrates this. We  can consider $C$ as a hyperbolic plane in the Poincar\'e disc model. With this interpretation, $A$ and $B$ represent perpendicular lines in this plane, and $C$ is the set of points at infinity.

Let $X$ be a line that divides $C$ into two arcs that each have probability $1/2$. There exists a  (combinatorially) unique line $Y$ perpendicular to $X$ that also divides $C$ into two arcs with probability $1/2$. The four arcs formed by the crossings of $C$ with both $X$ and $Y$ necessarily have probabilities $1/4+x,1/4-x,1/4+x,1/4-x$ for some $x$, but it will not necessarily be the case that $x=0$. Now, we conceptually rotate $X$ and $Y$, keeping them perpendicular and maintaining invariant the property that each of $X$ and $Y$ divides $P$ into two equal-probability arcs. As we do so, $x$ will change continuously; by the time we rotate $X$ into the position initially occupied by $Y$, $x$ will have negated its original value. Therefore, by the intermediate value theorem, there must be some position during the rotation at which $x=0$. The circles $A$ and $B$ formed by extending $X$ and $Y$ outside the model of the hyperbolic plane, for this position, satisfy the statement of the lemma.
\end{proof}

\begin{proof}[of Lemma~\ref{lem:quadrants}.]
For any sufficiently small number $\epsilon$, let $\Pi_\epsilon$ be the smooth probability distribution formed by adding a uniform distribution with total probability $\epsilon$ on all of $C$ to a uniform distribution with total probability $1-\epsilon$ on the points within distance $\epsilon$ of $P$.
Apply Lemma~\ref{lem:smooth} to $\Pi_\epsilon$, and let $A$ and $B$ be pairs of circles obtained in the limit as $\epsilon$ goes to zero. Then (if points on the boundaries of the arcs are assigned fractionally to the two arcs they bound as appropriate) the number of points assigned to each of the four arcs of $C$ disjoint from $A$ and $B$ is exactly $n/4$.

Next, rotate $A$ and $B$ by a small amount around their crossing point (as hyperbolic lines, that is) preserving their perpendicularity to each other and to $C$. This rotation causes them to become disjoint from all points in $P$. Each of the four arcs determined by the four crossing points, and each of the two longer arcs determined by two of the four crossing points, gains or loses only a fractional point by this rotation, so the inequalities $\lfloor n/4\rfloor\le n_i\le\lceil n/4\rceil$ and $\lfloor n/2\rfloor\le n_i+n_{(i+1)\bmod 4}\le\lceil n/2\rceil$ (where $n_i$ denotes the size of the $i$th arc) remain true after this rotation. However, there may be more than one solution to this system of inequalities, so we analyze cases according to the value of $n$ modulo four to determine that the solution obtained geometrically in this way matches the values of $n_i$ given to us in the lemma:
\begin{itemize}
\item If $n=0\pmod{4}$, the only choice for the values of $n_i$ is that all of them are equal to $n/4$.
\item If $n=1\pmod{4}$, then three of the $n_i$ must be $\lfloor n/4\rfloor$ and one must be $\lceil n/4\rceil$. By exchanging the roles of $A$ and $B$ as necessary we can ensure that the quadrant that is supposed to contain the larger number of points is the one that actually does.
\item If $n=2\pmod{4}$ then the only solution to the inequalities is that two opposite quadrants have $\lfloor n/4\rfloor$ points and the other two have $\lceil n/4\rceil$. Again, by exchanging $A$ and $B$ if necessary we can ensure that the correct two quadrants have the larger number of points.
\item If $n=3\pmod{4}$, then one of the $n_i$ must be $\lfloor n/4\rfloor$ and the remaining three must be $\lceil n/4\rceil$. Again, by exchanging the roles of $A$ and $B$ as necessary we can ensure that the quadrant that is supposed to contain the smaller number of points is the one that actually does.
\end{itemize}
Thus, in each case the partition satisfies the requirements of the lemma.
\end{proof}
\else
We defer the proof to the full version of the paper~\cite{XXX}.
\fi

\begin {lemma} \label {lem:manypoints}
  Given a circle $C$ and a set $P$ of $n$ points on $C$, there exists a point $p$ in $C$ such that we can draw $n$ edges from $p$ to the points in $P$ as circular arcs that lie completely inside $C$, do not cross each other, and meet in $p$ at $360/n^\circ$ angles.
\end {lemma}

\begin {proof}
Draw $n$ ports around a point with equal angles, and draw two perpendicular lines through the point (not coinciding with any ports), and count the number of points in each quadrant. Let these numbers be $n_1, \ldots, n_4$ and find two circles $A$ and $B$ as in Lemma~\ref {lem:quadrants}. Then we place $p$ at their intersection point inside $C$. Now orient the ports at $p$ such that each quadrant has the correct number of ports.

Within any quadrant, there is a circular arc tangent to $C$ at the point where it is crossed by $B$, and tangent to $A$ at point $p$; this can be seen by using a M\"obius transformation to transform $A$ and $B$ into a pair of perpendicular lines, after which the desired arc has half the radius of $C$. By the intermediate value theorem, there are two circular arcs from $p$ to any point $q$ on the boundary arc of the quadrant that remain entirely within the quadrant and are tangent to $A$ and $B$ respectively. By a second application of the intermediate value theorem, there is  a unique circular arc that connects $p$ to each connection point on the boundary of $C$, such that the outgoing direction at $p$ matches the port, and such that the arc remains entirely within its quadrant.

Any two arcs that belong to the same quadrant belong to two circles that cross at $p$ and at one more point. Whether that second crossing point is inside or outside of the quadrant can be determined by the relative ordering of the two arcs at $p$ and on the boundary of the quadrant. However, since the ordering of the ports and of the connection points is the same, none of the crossings of these circles are within the quadrant, so no two arcs cross.
\end {proof}

Figure~\ref {fig:plalom2-edges} illustrates the lemma.

\begin {theorem}
\label{thm:pointed-2L}
Every planar graph has a planar pointed $2$-Lombardi drawing. 
\end {theorem}

\begin {proof}
As in the previous section, we first obtain a touching-circles representation of a the given graph $G$ using the Koebe--Andreev--Thurston theorem. Each vertex $v$ in $G$ is represented by a circle $C$; place $v$ together with arcs connecting it to the set of contact points on $C$ using Lemma~\ref {lem:manypoints}. The arcs meet up at the contact points to form (non-smooth) $2$-Lombardi edges.
\end{proof}

\subsection{Smooth $3$-Lombardi planar realization for planar graphs}

Note that the $2$-Lombardi planar realization of the previous section has non-smooth bends in each edge. As we now show, every planar graph also has a smooth $3$-Lombardi drawing.

It seems likely that every planar graph $G$ has a smooth $3$-Lombardi drawing  formed  by perturbing each edge of a straight-line drawing of $G$ into a curve formed by two very small circular arcs near each endpoint of the edge, connected to each other by a straight segment. However, the details of this construction are messy.
An alternative construction is much simpler, once Theorem~\ref{thm:pointed-2L} is available.

\begin{theorem}
Every planar graph has a planar smooth $3$-Lombardi drawing. 
\end {theorem}

\begin{proof}
Find a pointed planar $2$-Lombardi drawing by Theorem~\ref{thm:pointed-2L}.
For each pointed bend of the drawing formed by two circular arcs $a_1$ and $a_2$, replace the bend by a third circular arc tangent to both $a_1$ and $a_2$, with the two points of tangency close enough to the bend to avoid crossing any other edge.
\end{proof}

\section{Conclusions}
We have proven several new results about planarity of Lombardi drawings and about classes of graphs that can be drawn with $k$-Lombardi drawings rather than $1$-Lombardi drawings.
However, several problems remain open, including the following:
\begin{enumerate}
\item Characterize the subclass of planar graphs that have 1-Lombardi planar realizations.
\item Characterize the subclass of planar graphs that have smooth 2-Lombardi planar realizations.
\item Bound the (change in) curvature of edge segments in $k$-Lombardi drawings.
\item Address area and resolution requirements for Lombardi drawings of graphs.
\end{enumerate}

\paragraph*{\bf Acknowledgments.}
This research was supported in part by the National Science
Foundation under grants CCF-0830403, CCF-0545743, and CCF-1115971,
by the
Office of Naval Research under MURI grant N00014-08-1-1015,
and by the Louisiana Board of Regents through 
PKSFI Grant LEQSF (2007-12)-ENH-PKSFI-PRS-03.

{ 
\raggedright
\bibliographystyle{abuser}
\bibliography{lombardi,lombardi2}
}

\end{document}